\documentclass[a4paper,11pt,fleqn]{article}

\usepackage{amsmath,amssymb,amsthm}
\usepackage[mathscr]{eucal}
\usepackage{graphicx}
\usepackage{xcolor}
\usepackage{eufrak,mathrsfs,longtable}

\newcommand{\gen}[1]{\partial_{#1}}
\newcommand{\pr}[1]{\rm pr^{(#1)}}
\newcommand{\todo}[1][\null]{\ensuremath{\clubsuit}}
\makeatletter
\long\def\@makecaption#1#2{%
  \vskip\abovecaptionskip\footnotesize
  \sbox\@tempboxa{#1. #2}%
  \ifdim \wd\@tempboxa >\hsize
    #1. #2\par
  \else
    \global \@minipagefalse
    \hb@xt@\hsize{\hfil\box\@tempboxa\hfil}%
  \fi
  \vskip\belowcaptionskip}
\makeatother

\newtheorem{theorem}{Theorem}

\newtheorem{corollary}[theorem]{Corollary}
\newtheorem{proposition}[theorem]{Proposition}
\newtheorem*{problem*}{Problem}
{\theoremstyle{definition}
\newtheorem{definition}[theorem]{Definition}

\newtheorem{remark}[theorem]{Remark}
\newtheorem*{remark*}{Remark}
}

\begin{document}

\par\noindent {\LARGE\bf
Abelian Lie symmetry algebras of two-dimensional quasilinear evolution equations
\par}

\vspace{4mm}\par\noindent {\large Rohollah Bakhshandeh-Chamazkoti$^{\dag}$
} \par\vspace{2mm}\par

\vspace{2mm}\par\noindent {\it
$^{\dag}$~Department of Mathematics, Faculty of Basic Sciences,
Babol  Noshirvani University of Technology, Babol, Iran.\par
}

\vspace{2mm}\par\noindent
\textup{E-mail:} r\_bakhshandeh@nit.ac.ir\!
\par

\vspace{8mm}\par\noindent\hspace*{10mm}\parbox{140mm}{\small
We carry out the classification of abelian Lie symmetry algebras of
two-dimensional second-order nondegenerate quasilinear evolution equations.
It is shown that such an equation is linearizable if it admits an abelian Lie symmetry algebra that is of dimension greater
than or equal to five or of dimension greater than or equal to three with rank one.}\par\vspace{4mm}

\section{Introduction}
Transformation properties of evolution equations especially Lie group
of point symmetries has been widely studied
because of many practical benefits that such knowledge gives and also because of the variety
of physical applications that are modeled by these equations.

The purpose of this paper is  to classify all admissible abelian symmetry algebras
of the class ${\mathcal L}$ of two-dimensional  second-order non-degenerate quasilinear evolution equations, which are of the general form
\begin{gather}\label{eveqn}
\begin{split}
u_t={}&F(t, x, y, u, u_x, u_y)u_{xx}+G(t, x, y, u, u_x, u_y)u_{yy}+H(t, x, y, u, u_x, u_y)u_{xy}
\\{}&+K(t, x, y, u, u_x, u_y),
\end{split}
\end{gather}
where $F, ~G, ~H$ and $K$ are some non-zero smooth functions in terms of $t, x, y, u, u_x, u_y$
variables with $FG-H^2/4\neq 0$. The present paper is the  first part of a project to
classify Lie symmetry algebras of differential equations of the class ${\mathcal L}$ and
 develops  the methods which were applied in \cite{Basarab-Horwath2013}.

In \cite{Basarab-Horwath2013}, the authors gave a complete Lie point symmetry classification of all third-order evolution
equations
\begin{align}\label{peq}
u_t=F(t, x, u, u_x, u_{xx})u_{xxx}+G(t, x, u, u_x, u_{xx}),
\end{align}
 where $F\neq 0$ and $G$ are two  functions in terms of $t, x, u, u_x, u_{xx}$  which \eqref{peq} admit semi-simple symmetry algebras.
 Before in \cite{Gungor2004}, this method was applied for equation \eqref{peq}
 for $F=1$. This method for finding symmetries and integrability properties for $(1+1)$-dimension evolution equations
  were carried out in \cite{Basarab-Horwath2001, Gazeau92, Gungor96, Gungor2004, Lahno05, Zhdanov99, Zhdanov07},
 and recently this method is developed in
 \cite{Basarab-Horwath2017, Popovych2017, Kurujyibwami2018, Kontogiorgis2018, Vaneeva2020}.

In \cite{Kontogiorgis2018}, the authors carried out enhanced symmetry analysis of the two-dimensional
Burgers system
\begin{align}\label{sysburg}
\left\{ \begin{array}{lcl}
u_t + uu_x +vu_y - u_{xx} - u_{yy} = 0,\\
 v_t +uv_x + vv_y - v_{xx} - v_{yy} = 0,
\end{array}\right.
\end{align}
Under the constraint $v = 0$, the second equation of the system \eqref{sysburg} is equaled to zero
and its first equation reduces to a $(1 + 2)$-dimensional generalization of the Burgers
equation
\begin{align}\label{eqburg}
u_t + uu_x - u_{xx} - u_{yy} = 0,
\end{align}
which was deduced in \cite{Rajaee2008} as an equation for the wave phase of two-dimensional sound
simple waves in weakly dissipative flows. Therein, symmetry analysis of this equation was
carried out, which included the first exhaustive study of its Lie reductions in an optimized
way and the construction of several new families of its exact solutions.
The group classification of $(1+2)$-dimensional diffusion-convection equation \eqref{eqburg} was also investigated  in \cite{Demetriou2008}.
Recently, Bihlo and Popovych in \cite{Popovych2020}, classified zeroth-order conservation laws of systems from the class of two-dimensional
shallow water equations with variable bottom topography.

The present paper is organized as follows: Section 2 is presented to some basic definitions
of equivalence groupoid and then we show the infinitesimal Lie point symmetry of differential equations of the class ${\mathcal L}$,
has the form
$$Q=a(t)\gen t +b(t, x, y, u)\gen x+c(t, x, y, u)\gen y +\varphi(t, x, y, u)\gen u, $$
where $a(t),\; b(t, x, y, u),\; c(t, x, y, u)$ and $\varphi(t, x, y, u)$ are
smooth functions of their arguments which $Q$ is equal to
$\gen t$ or $\gen u$
as a canonical form under the equivalence group \eqref{eqgrp} and also the $\gen x, \gen y$ and $\gen u$ are equivalent
as canonical forms under  this equivalence group  \eqref{eqgrp}.

In section 3, we compute the determining equations by applying Lie symmetry conditions.

In section 4, we classify the admissible abelian symmetry algebras and we prove that any equation
of the the class ${\mathcal L}$ equipped with an five-dimensional abelian symmetry algebra or with dimension bigger than five, or
three-dimensional abelain Lie algebra or with dimension bigger than three with rank-one realization, is linearizable.
Finally, a table included some invariant solutions corresponding to abelian Lie  symmetries for non-linear cases of
the evolution equation is listed.

%%%%%%%%%%%%%%%%%%%%%%%%%%%%%%%%%%%%%%%%%%%%%%%%%%%%%%%%%%%%%
\section{Equivalence groupoid}
Let  ${\mathcal L}_\theta$ denote a system of differential equations of the form $L(x,u^{(r)},\theta^{(q)}(x,u^{(r)}))=0$, where $x, u$,
and $u^{(r)}$  are the tuples of independent variables, of dependent variables and of derivatives of $u$
with respect to $x$ up to order $r$. The tuple of functions $L = (L^1,\dots,L^l)$ of $(x,u^{(r)}, \theta)$ is fixed
whereas the tuple of functions $\theta = (\theta^1,\dots,\theta^k)$ of $(x,u^{(r)})$ runs through the solution set~$\mathcal S$ of
an auxiliary system of differential equations and inequalities in $\theta$, where $x$ and $u^{(r)}$ jointly play
the role of independent variables. Thus, the class of (systems of) differential equations ${\mathcal L}$ is the
parameterized family of systems ${\mathcal L}_\theta$ with $\theta $ running through the set~$\mathcal S$. The components of $\theta$ are
called the arbitrary elements of the class ${\mathcal L}$.
The equivalence groupoid~$\mathcal G^\sim$ of the class ${\mathcal L}$ consists of
admissible transformations of this class,
$$\mathcal G^\sim=\mathcal G^\sim(\mathcal L) = \big\{(\theta,\xi,\tilde\theta)\mid\theta,\tilde\theta\in\mathcal S, \xi \in \mathrm T(\theta,\tilde\theta)\big\}.$$
Here  ${\mathcal L}_\theta$ and  ${\mathcal L}_{\tilde{\theta}}$ are
the source and the target systems, belonging to the class ${\mathcal L}$ and corresponding to the values $\theta$
and $\tilde{\theta}$ of the arbitrary-element tuple, respectively, which in turn runs through the solution set~$\mathcal S$ of
the auxiliary system $S$ of equations and inequalities for arbitrary elements,
and $\xi$ is a point transformation relating the equations ${\mathcal L}_\theta$ and ${\mathcal L}_{\tilde{\theta}}$
(the set of all such transformations is denoted by $\mathrm T(\theta,\tilde\theta)$) \cite{Bihlo2012}.

%\todo Definition of (usual) equivalence group!
\begin{definition}\label{def:UsualEquivGroup}
The \emph{(usual) equivalence group} $G^{\sim}=G^{\sim}(\mathcal L)$ of the class~$\mathcal L$
is the (pseudo)group of point transformations in the space of $(x,u^{(r)},\theta)$ which are projectable to the space of $(x,u^{(r')})$ for any $0\le r'\le r$,
are consistent with the contact structure on the space of $(x,u^{(r)})$ and preserve the  class~$\mathcal L$.
\end{definition}

\begin{definition}\label{def:Normalized}
The class of differential equations~${\mathcal L}$ is called \emph{normalized}
if its equivalence groupoid~$\mathcal G^\sim$ is induced by its equivalence group~$G^\sim$,
meaning that for any triple $(\theta,\xi,\tilde\theta)$ from~$\mathcal G^\sim$,
there exists a transformation $\Xi$ from $G^\sim$ such that $\tilde\theta =\Xi_*\theta$ and $\xi=\Xi|_{(x,u)}$.
\end{definition}

For the class~$\mathcal L$ of equations of the form~\eqref{eveqn},
the tuple~$\theta$ of arbitrary elements is constituted by four arbitrary smooth functions~$F$, $G$, $H$ and~$K$ of $(t,x,y,u,u_t,u_x,u_y)$
that satisfy the auxiliary system~$S$ consisting of the single inequality $FG-H^2/4\neq 0$,
that is, $\theta=(F,G,H,K)$ and $\mathcal S=\{(F,G,H,K)\mid FG-H^2/4\neq 0\}$.
The usual equivalence group~$\mathcal G^\sim$ of the class~$\mathcal L$ can be assumed to act in the space with coordinates $(t,x,y,u,u_t,u_x,u_y,F,G,H,K)$
and to consist of the point transformations of the form
\begin{gather}\label{geeqgr}
\begin{split}
&t'=T(t,x,y,u),\quad
 x'=X(t,x,y,u),\quad
 y'=Y(t,x,y,u),\quad
 u'=U(t,x,y,u),\\
&u'_{t'}=U^t(t,x,y,u,u_x,u_y),\quad
 u'_{x'}=U^x(t,x,y,u,u_x,u_y),\quad
 u'_{y'}=U^y(t,x,y,u,u_x,u_y),\\
&F'=F(t,x,y,u,u_x,u_y,F,G,H,K),\quad
 G'=G(t,x,y,u,u_x,u_y,F,G,H,K),\\
&H'=H(t,x,y,u,u_x,u_y,F,G,H,K),\quad
 K'=K(t,x,y,u,u_x,u_y,F,G,H,K),
\end{split}
\end{gather}
that preserve the contact structure of the space with coordinates $(t,x,y,u,u_x,u_y)$
and map each equation~$\mathcal L_\theta$ from the class~${\mathcal L}$ to an equation~$\mathcal L_{\theta'}$ from the same class,
\begin{gather}\label{eveqnnew}
\begin{split}
u'_{t'}={}&F'(t',x',y',u',u'_{x'},u'_{y'})u'_{x'x'}+G'(t',x',y',u',u'_{x'},u'_{y'})u'_{y'y'}+H'(t',x',y',u',u'_{x'},u'_{y'})u'_{x'y'}\\
&{}+K'(t',x',y',u',u'_{x'},u'_{y'}),
\end{split}
\end{gather}

\begin{theorem}\label{pointequiv}
A point transformation~$\xi$ in the space with coordinates $(t, x, y, u)$ transforms
an equation~\eqref{eveqn} from the class $\mathcal L$
to an equation \eqref{eveqnnew} from the same class if and only if the components of $\xi$ are of the form
\begin{equation}\label{eqgrp}
t'=T(t),\quad x'=X(t, x, y, u),\quad y'=Y(t, x, y, u), \quad u'=U(t, x, y ,u).
\end{equation}
with the condition $T_t\displaystyle{\left|\frac{\partial(X, Y, U)}{\partial(x, y, u)}\right|\neq0}$.
The transformed arbitrary elements $F'$, $G'$, $H'$ and $K'$ are given  by
%\begin{subequations}\label{FGHK}
\begin{eqnarray} \label{F}
 F'&=&- \frac{1}{T_t\zeta_u^2}\Big[(X_x\zeta_u-X_u\zeta_x)^2~F+(X_y\zeta_u-X_u\zeta_y)^2~G+(X_x\zeta_u-X_u\zeta_x)(X_y\zeta_u-X_u\zeta_y)~ H\Big]\\ \label{G}
 G'&=&- \frac{1}{T_t\zeta_u^2}\Big[(Y_x\zeta_u-Y_u\zeta_x)^2~ F  + (Y_y\zeta_u-Y_u\zeta_y)^2 ~G + (Y_x\zeta_u-Y_u\zeta_x)(Y_y\zeta_u-Y_u\zeta_y)~H\Big]\\\nonumber
 H'&=&- \frac{2}{T_t\zeta_u^2}\Big[(X_x\zeta_u-X_u\zeta_x)(Y_x\zeta_u-Y_u\zeta_x)~ F -  (X_y\zeta_u-X_u\zeta_y)(Y_y\zeta_u-Y_u\zeta_y)~G
-(X_x\zeta_u\\\ \label{H}
&& \hspace*{1.3cm}-X_u\zeta_x)(Y_y\zeta_u-Y_u\zeta_y)~H\Big]\\\nonumber
K'&=&\frac{1}{T_t}\Big[\Big(-\zeta_{xx}+\frac{\zeta_{x}^2}{\zeta_u} -\frac{\zeta_{x}^2\zeta_{uu}}{\zeta_u^2}\Big)~F  + \Big(-\zeta_{yy}+\frac{\zeta_{y}^2}{\zeta_u} -\frac{\zeta_{y}^2\zeta_{uu}}{\zeta_u^2}\Big)~G+\Big(-\zeta_{xy}+\frac{\zeta_{x}\zeta_{y}}{\zeta_u} -\frac{\zeta_{x}\zeta_{y}\zeta_{uu}}{\zeta_u^2}\Big)~H\\\label{K}
&& \hspace*{1.3cm}+\frac{K}{\zeta_u}+\frac{\zeta_t}{\zeta \zeta_u }\Big]
\end{eqnarray}
%\end{subequations}
where $\zeta={\partial u}=U-u_{t}T-u_{x}X-u_{y}Y$.
\end{theorem}
\begin{proof}
We employ the direct method for finding point transformations relating differential equations.
Suppose that a point transformation~$\xi$ of the general form
\begin{equation}\label{eqgrp_gen}
t'=T(t,x,y,u),\quad
x'=X(t,x,y,u),\quad
y'=Y(t,x,y,u),\quad
u'=U(t,x,y,u)
\end{equation}
with $J:=\big|\partial(T,X,Y,U)/\partial(t,x,y,u)\big|\neq0$ maps the equation~\eqref{eveqn} to the equation~\eqref{eveqnnew}.

Computing $u_x$, $u_y$ and $u_t$, we get
\begin{eqnarray}\label{umu}
u_\mu=-\frac{\partial_\mu u'}{\partial_u u'}=-\frac{U_\mu-u'_{t'}T_\mu-u'_{x'}X_\mu-u'_{y'}Y_\mu}{U_u-u'_{t'}T_u-u'_{x'}X_u-u'_{y'}Y_u}
\end{eqnarray}
where $\mu$ runs through $\{t,x,y\}$.
Now, making the change of variables \eqref{geeqgr}, we get  the derivatives $u_{xx}, u_{yy}, u_{xy}$ by the following
relations
\begin{eqnarray}\label{umunu}
u_{\mu\nu}=-\frac{(\partial_{\mu\nu} u')(\partial_u u')-(\partial_{u\nu}u')(\partial_\mu u')}{\partial_u u'}
\end{eqnarray}
where $\mu, \nu$ could be equal to $t, x$ and $y$.
Then, we insert $u_t, u_{xx}, u_{yy}, u_{xy}$ from \eqref{umu} and \eqref{umunu} into \eqref{eveqn},
and then in the transformed equation, the coefficients  $u'_{t'x'},$ $u'_{t'y'}$ and $u'_{t't'}$ are set equal to zero
because by comparing the transformed equation with equation \eqref{eveqnnew},
the right-hand side of the  transformed equation must not contain $u'_{t'x'},$ $u'_{t'y'}$ and $u'_{t't'}$.
Therefore
\begin{eqnarray}
&&\hspace*{-9ex}F(T_x {\partial_u u'}-T_u {\partial_x u'})^2+G(T_y {\partial_u u'} - T_u {\partial_y u'})^2+H(T_x {\partial_u u'}-T_u {\partial_x u'})(T_y {\partial_u u'} - T_u {\partial_y u'})=0,
\\[5mm]
&&\hspace*{-6ex}\begin{split}
&F(T_x {\partial_u u'}-T_u {\partial_x u'})(X_x {\partial_u u'}-X_u {\partial_x u'})+G(T_y {\partial_u u'} - T_u {\partial_y u'})(X_y {\partial_u u'} - X_u {\partial_y u'})\\
&\qquad+H(T_x {\partial_u u'}-T_x {\partial_x u'})(X_y {\partial_u u'} - X_u {\partial_y u'})=0,
\end{split}
\\[5mm]
&&\hspace*{-6ex}\begin{split}
&F(T_x {\partial_u u'}-T_u {\partial_x u'})(Y_x {\partial_u u'}-Y_u {\partial_x u'})+G(T_y {\partial_u u'} - T_u {\partial_y u'})(Y_y {\partial_u u'} - Y_u {\partial_y u'})\\
&\qquad+H(T_x {\partial_u u'}-T_x {\partial_x u'})(Y_y {\partial_u u'} - Y_u {\partial_y u'})=0,
\end{split}
\end{eqnarray}
According to the above system can be written
\begin{gather*}
\begin{split}
&{\rm rank}\left[\begin{array}{ccccc}
    T_x {\partial_u u'}-T_u {\partial_x u'} & T_y {\partial_u u'} - T_u {\partial_y u'} &  T_y {\partial_u u'} - T_u {\partial_y u'}\\
   X_x {\partial_u u'}-X_u {\partial_x u'} & X_y {\partial_u u'} - X_u {\partial_y u'} & X_y {\partial_u u'} - X_u {\partial_y u'}\\
   Y_x {\partial_u u'}-Y_u {\partial_x u'} & Y_y {\partial_u u'} - Y_u {\partial_y u'} & Y_y {\partial_u u'} - Y_u {\partial_y u'}
      \end{array}\right]
\\[1ex]
&\qquad  {}= {\rm rank}\left[\begin{array}{ccccc}
    T_x  & T_y &   T_u \\
   X_x  & X_y  & X_u \\
    Y_x & Y_y  & Y_u \\
    \partial_{x}u' & \partial_{y}u' & \partial_{u}u'
      \end{array}\right]-1
     ={\rm rank}\left[\begin{array}{ccccc}
    T_x  & T_y &   T_u \\
   X_x  & X_y  & X_u \\
    Y_x & Y_y  & Y_u \\
    U_x & U_y & U_u
      \end{array}\right]-1
   =3,
\end{split}
\end{gather*}
since $J\neq0$. Now since  the matrix
$\left[\begin{array}{cc}
    F & H/2 \\
    H/2 & G \\
  \end{array}\right]$
is nondegenerate then we conclude
\begin{eqnarray}\label{eqfi}
\left\{\begin{array}{cccccc}
    T_x {\partial_u u'}-T_u {\partial_x u'} =0,\\[3mm]
     T_y {\partial_u u'} - T_u {\partial_y u'}=0, \\[3mm]
      T_y {\partial_u u'} - T_u {\partial_y u'}=0.
      \end{array}\right.
  \end{eqnarray}
Splitting the variables  in \eqref{eqfi} we find
$T_\mu U_u-T_u U_\mu=0$ and
\begin{eqnarray}
    T_\mu T_u =0,\qquad
     T_\mu X_u - T_u X_\mu=0, \qquad
      T_\mu Y_u - T_u Y_\mu=0.
  \end{eqnarray}
where $\mu=x, y$.  If $T_u \ne0$ then each row
$(X_x, X_y, X_u)$, $(Y_x, Y_y, Y_u)$ and $(U_x, U_y, U_u)$  are parallel to the $(T_x, T_y, T_u)$
which contradicts with $J\neq0$, and therefore $T_u=0$. Using $J\neq0$ condition we conclude
$(X_u, Y_u, U_u)\neq(0, 0, 0)$ and $T_\mu U_u=T_\mu X_u=T_\mu Y_u=0$ here $\mu$ are $x, y$.
So $T_x=T_y=0$.

Putting $\zeta={\partial u}=U-u_{t}T-u_{x}X-u_{y}Y$, and using the following formula
\begin{eqnarray}
u_t=-\dfrac{\partial_t u'}{\partial_u u'}=-\dfrac{\zeta_t}{\zeta_u}+\dfrac{T_t}{\zeta_u}u'_{t'}
\end{eqnarray}
we have
\begin{eqnarray*}
&&F\Big[-\frac{\zeta_{xx}}{\zeta_u}+\frac{\zeta_{x}^2}{\zeta_u^2} -\frac{\zeta_{x}^2\zeta_{uu}}{\zeta_u^3}-\frac{1}{\zeta_u^3}\big [(X_x\zeta_u-X_u\zeta_x)^2u'_{x'x'}
+(Y_x\zeta_u-Y_u\zeta_x)^2u'_{y'y'}+2(X_x\zeta_u-X_u\zeta_x)\\
&&~~(Y_x\zeta_u-Y_u\zeta_x)u'_{x'y'}\big]\Big] + G\Big[-\frac{\zeta_{yy}}{\zeta_u}+\frac{\zeta_{y}^2}{\zeta_u^2} -\frac{\zeta_{y}^2\zeta_{uu}}{\zeta_u^3}   -\frac{1}{\zeta_u^3}\big [(X_y\zeta_u-X_u\zeta_y)^2u'_{x'x'}+(Y_y\zeta_u-Y_u\zeta_y)^2u'_{y'y'}\\
&&~~+2(X_y\zeta_u-X_u\zeta_y)(Y_y\zeta_u-Y_u\zeta_y)u'_{x'y'}\big]\Big]+H\Big[-\frac{\zeta_{xy}}{\zeta_u}+\frac{\zeta_{x}\zeta_{y}}{\zeta_u^2} -\frac{\zeta_{x}\zeta_{y}\zeta_{uu}}{\zeta_u^3}   -\frac{1}{\zeta_u^3}\big [(X_x\zeta_u-X_u\zeta_x)\\
&&~~ (X_y\zeta_u-X_u\zeta_y)u'_{x'x'}+(Y_x\zeta_u-Y_u\zeta_x)(Y_y\zeta_u-Y_u\zeta_y)u'_{y'y'}+2(X_x\zeta_u-X_u\zeta_x)(Y_y\zeta_u-Y_u\zeta_y)u'_{x'y'}\big]\Big]\\
&&~~ +K=-\frac{\zeta_t}{\zeta}+\frac{T_t}{\zeta_u}\big(F'u'_{x'x'}+G'u'_{y'y'}+H'u'_{x'y'}+ K'\big)
\end{eqnarray*}
and then, splitting the last equation with respect to~$u'_{x'x'}$, $u'_{x'y'}$ and $u'_{y'y'}$, we obtain the equations~\eqref{F}, \eqref{G}, \eqref{H} and \eqref{K}.
\end{proof}

In other words, Theorem~\ref{pointequiv} means
that the equivalence groupoid~$\mathcal G^\sim$ of the class ${\mathcal L}$ consists of
of the triples $(\theta,\xi,\tilde\theta)$, where the point transformations~$\xi$ is of the form~\eqref{eqgrp},
the tuples~$\theta$ runs through the entire set~$\mathcal S$, and
the tuples~$\tilde\theta$ are defined by the equations~\eqref{F}--\eqref{K}.
Thus, the equivalence groupoid~$\mathcal G^\sim$ can be partitioned into families of admissible transformations
each of which corresponds to a fixed point transformations~$\xi$ and parameterized by arbitrary values of~$\theta\in\mathcal S$.
This means that the each point transformation of the form~\eqref{eqgrp} is the projection of an equivalence transformation
to the space with coordinates $(t,x,y,u)$, the $(F,G,H,K)$-components of this equivalence transformation
are defined by the equations~\eqref{F}, \eqref{G}, \eqref{H} and \eqref{K}, respectively, and
the action groupoid of the equivalence group~$G^\sim$ of the class ${\mathcal L}$
coincides with the entire equivalence groupoid~$\mathcal G^\sim$ of this class.
As a result, we have the following theorem.

\begin{theorem}
The class${\mathcal L}$ of equations of the form~\eqref{eveqn} with $FG-H^2/4\neq 0$ is normalized in the usual sense.
Its usual equivalence group~$G^\sim$ is constituted by the point transformations in the space with coordinates $(t,x,y,u,u_t,u_x,u_y,F,G,H,K)$,
where
the $(t,x,y,u)$-components are given by~\eqref{eqgrp},
the $(u_t,u_x,u_y)$-components are computed using the chain rule, and
the $(F,G,H,K)$-components are defined by the equations~\eqref{F}, \eqref{G}, \eqref{H} and \eqref{K}, respectively.
\end{theorem}

%--------------------------------------------
\begin{corollary}
A smooth vector field $Q$ which defines an infinitesimal point symmetry of equation \eqref{eveqn}, has the form
\begin{eqnarray}\label{vecinf}
Q=a(t)\gen t +b(t, x, y, u)\gen x+c(t, x, y, u)\gen y +\varphi(t, x, y, u)\gen u.
\end{eqnarray}
where $a(t),\; b(t, x, y, u),\; c(t, x, y, u),\; \varphi(t, x, y, u)$ are smooth functions of their arguments.
\end{corollary}
%
%-----------------------------------------------
%
\begin{theorem}\label{canonicalform}
A vector field
\[ Q=a(t)\gen t + b(t, x, y, u)\gen x+ c(t, x, y, u)\gen y + \varphi(t, x, y, u)\gen u \]
can be transformed by transformations of the form (\ref{eqgrp}) into one of the following canonical forms:
\begin{equation}
Q=\gen t,\qquad Q=\gen u.
\end{equation}
\end{theorem}
\begin{proof}
Any operator
\[ Q=a(t)\gen t + b(t, x, y, u)\gen x+ c(t, x, y, u)\gen y + \varphi(t, x, y, u)\gen u, \]
is transformed by our allowed transformations into $Q'=Q(T)\gen {t'} + Q(X)\gen {x'} + Q(Y)\gen {y'} + Q(U)\gen {u'}$. If $a(t)\neq 0$ then we
choose $T(t)$ so that $Q(T)=a(t)\dot{T}(t)=1$ (at least locally) and also we can choose $X, ~Y$ and $U$ to be each of any three
independent integrals of the PDE $Q(V)=0$. This gives the canonical form $Q=\gen t$ in some coordinate
system.
If we now have $a(t)=0$ then $Q$ is transformed into $Q'=Q(X)\gen {x'} + Q(Y)\gen {y'} + Q(U)\gen {u'}$. We then choose $X, ~Y$ and $U$ so
that $ Q(X)=Q(Y)=0,\;\; Q(U)=1$. This gives us the canonical form $Q=\gen u$ in some coordinate system.
\end{proof}
\begin{remark}
It is sometimes useful to use the fact that $t'=t,\;\, x'=y'=u,\;\; u'=x=y$ is an equivalence
transformation preserving the form of equation (\ref{eveqn}), and it gives us the following transformations of
derivatives:
\[ \gen t\leftrightarrow \gen t, \quad \gen x \leftrightarrow \gen y, \quad \gen x \leftrightarrow \gen u, \quad \gen y \leftrightarrow \gen u. \]
\end{remark}
%
%------------------------------------------------------------------------------------------------------------
%
\section{The symmetry condition}
The second prolongation of $Q$ in \eqref{vecinf}
is the vector field
\begin{gather}\label{pr2}
\begin{split}
{\pr{2}}Q&{}=Q+\varphi^x {\partial_{u_x}}+\varphi^y {\partial_{u_y}}+\varphi^t {\partial_{u_t}}+
\varphi^{xx}{\partial_{u_{xx}}}+\varphi^{xy}{\partial_{u_{xy}}}
+\varphi^{xt}{\partial_{u_{xt}}}+\varphi^{yy}{\partial_{u_{yy}}}\\
&{}~~+\varphi^{yt}{\partial_{u_{yt}}}+\varphi^{tt}{\partial_{u_{tt}}},
\end{split}
\end{gather}
with coefficients
\begin{eqnarray}
\varphi^{i}&=&D_{i}(\varphi- b u_x-c u_y-a u_t)+b u_{xi}+c u_{yi}+a u_{ti},\\\label{eq:8}
\varphi^{ij}&=&D_{i}(D_{j}(\varphi -b u_x-c u_y-a u_t)) +b u_{xij}+c u_{yij}+au_{tij},\label{eq:9}
\end{eqnarray}
where  $D_i$
represents total derivative and subscripts of $u$ are derivative
in terms of the respective coordinates $i$ and $j$
in above could be $x,y$ or $t$ coordinates.  Then the vector field $Q$
is a symmetry of \eqref{eveqn} precisely when its second-order prolongation \eqref{pr2} annihilates equation
\eqref{eveqn} on its solution manifold, that is we have
\begin{equation}\label{symcon1}
{\pr{2}}Q(\Delta)\Bigl|_{\Delta=0}=0,\quad \Delta=u_t-Fu_{xx}-Gu_{yy}-Hu_{xy}-K,
\end{equation}
namely
\begin{eqnarray}\nonumber\label{symcon2}
 &\Big[\varphi_t-[Q(F)+\varphi_x F_{u_x}+\varphi_y F_{u_y}]u_{xx}-[Q(G)+\varphi_x G_{u_x}+\varphi_y G_{u_y}]u_{yy}-[Q(H)+\varphi_x H_{u_x}\\
 & +\varphi_y H_{u_y}]u_{xy}-[Q(K)+\varphi_x K_{u_x}+\varphi_y K_{u_y}]
 -\varphi_{xx}F-\varphi_{yy}G-\varphi_{xy}H\Big|_{\Delta=0}=0.
\end{eqnarray}
\begin{proposition}\label{symmcondition}
The symmetry group of the nonlinear equation \eqref{eveqn} for
 functions $F,~G,~H$ and $K$ is generated by the vector field
\begin{equation}\label{gvf}
Q=a(t)\gen t + b(t, x, y, u)\gen x+ c(t, x, y, u)\gen y + \varphi(t , x, y, u)\gen u
\end{equation}
where the functions $a$, $b$, $c$ and $\varphi$ satisfy the determining equations
\begin{subequations}\label{determin}
\begin{align} \nonumber
& aF_t+bF_x+cF_y+\varphi F_u+[\varphi_x+(\varphi_u-b_x)u_x-c_xu_y-b_uu_x^2-c_uu_xu_y]F_{u_x}\\ \nonumber
& +[\varphi_y+(\varphi_u-c_y)u_y-b_yu_x-c_uu_y^2-b_uu_xu_y]F_{u_y}+[a_t-2b_x-2b_uu_x]F\\\label{sym1}
& -[b_y+b_uu_y]H=0,
\end{align}
\begin{align} \nonumber
& aG_t+bG_x+cG_y+\varphi G_u+[\varphi_x+(\varphi_u-b_x)u_x-c_xu_y-b_uu_x^2-c_uu_xu_y]G_{u_x}\\ \nonumber
& +[\varphi_y+(\varphi_u-c_y)u_y-b_yu_x-c_uu_y^2-b_uu_xu_y]G_{u_y}+[a_t-2c_y-2c_uu_y]G\\\label{sym2}
& -[c_x+c_uu_x]H=0,
\end{align}
\begin{align} \nonumber
& aH_t+bH_x+cH_y+\varphi H_u+[\varphi_x+(\varphi_u-b_x)u_x-c_xu_y-b_uu_x^2-c_uu_xu_y]H_{u_x}\\
\nonumber
& +[\varphi_y+(\varphi_u-c_y)u_y-b_yu_x-c_uu_y^2-b_uu_xu_y]H_{u_y}+[a_t-b_x-c_y-c_uu_y-b_uu_x]H\\
\nonumber
&-2(c_x+c_uu_x)F-2(b_y+b_uu_y)G=0,\\\label{sym4}
\end{align}
\begin{align} \nonumber
& aK_t+bK_x+cK_y+\varphi K_u+[\varphi_x+(\varphi_u-b_x)u_x-c_xu_y-b_uu_x^2-c_uu_xu_y]K_{u_x}\\
\nonumber
& +[\varphi_y+(\varphi_u-c_y)u_y-b_yu_x-c_uu_y^2-b_uu_xu_y]K_{u_y}+[a_t-\varphi_u+b_uu_x+c_uu_y]K\\
\nonumber
& +[\varphi_{xy}+(\varphi_{yu}-b_{xy})u_x+(\varphi_{xu}-c_{xy})u_y+(\varphi_{uu}-b_{xu}-c_{yu})u_xu_y-b_{yu}u_x^2\\
\nonumber
&-b_{uu}u_x^2u_y-c_{xu}u_y^2-c_{uu}u_xu_y^2]H+[\varphi_{yy}+(2\varphi_{yu}-c_{yy})u_y+(\varphi_{uu}-2c_{yu})u_y^2\\
\nonumber
& -c_{uu}u_y^3-b_{yy}u_x-2b_{yu}u_xu_y-b_{uu}u_xu_y^2]G+[\varphi_{xx}+(2\varphi_{xu}-b_{xx})u_x-c_{xx}u_y\\
\nonumber
&+(\varphi_{uu}-2b_{xu})u_x^2-b_{uu}u_x^3-2c_{xu}u_xu_y-c_{uu}u_x^2u_y]F-\varphi_t+b_t u_x+c_t u_y=0.\\\label{sym4}
\end{align}
\end{subequations}
\end{proposition}
%
%---------------------------------------------------------------------------------------------------------------------------------------------------
%
\section{Abelian Lie algebras as symmetries}
\begin{theorem}\label{abeliansymm1}
There are only two inequivalent forms for one-dimensional
abelian Lie symmetry algebras  $A=\langle \gen t\rangle$ and $A=\langle \gen u\rangle$.
\end{theorem}
%---------------------------------------------------------
\begin{theorem}\label{abeliansymm2}
All inequivalent, admissible two-dimensional abelian Lie symmetry algebras  are
$\langle \gen t, \gen u \rangle$, $\langle \gen x, \gen u\rangle$ and $\langle \gen u, x\gen u\rangle$.
\end{theorem}
\begin{proof}
Now we take $\dim A=2$. We have $A=\langle Q_1, Q_2\rangle$. For $\text{rank}\, A=1$  the case $Q_1=\gen t$ is impossible
because it leads to $Q_2=a(t)\gen t$ and $[Q_1,Q_2]=0$ gives $\dot{a}(t)=0$, so $\dim A=1$.
We take $Q_1=\gen u$ and $Q_2=\varphi(t,x,y)\gen u$. The residual equivalence group $G^\sim(\gen u)$ is given by
 $$t'=T(t),\, x'=X(t,x,y),\,y'=Y(t,x,y),\, u'=u + U(t,x,y)$$
where $\dot{T}(t)\left(X_xY_y-X_yY_x\right)\neq 0$. Under such a transformation $Q_2$ is mapped to $Q'_2=\varphi(t,x,y)\gen {u'}$. If $\varphi_x\neq 0$  (or $\varphi_y\neq 0$)
we may take
$X(t,x,y)=\varphi(t,x,y)$ so that the $Q'_2=x'\gen {u'}$ and then we have $\langle \gen u, x\gen u\rangle$ in canonical
form. If $\varphi_x=\varphi_y=0$ then $Q_2=\varphi(t)\gen u$. Substituting $\gen u$ and $\varphi(t)\gen u$ into the equation (\ref{sym2})  we find
$G_u=0,\; \dot{\varphi}(t)=0$ so $\varphi(t)$ is a constant and thus $\dim A=1$ which is a contradiction, so we have no
admissible canonical form in this case.

If $\dim A=2,\; \text{rank}\, A=2$ then we may take $Q_1=\gen t$ or $Q_1=\gen u$. If $Q_1=\gen t$ then, by
commutativity and the fact that  $A=\langle Q_1, Q_2\rangle$, we may take $Q_2=b(x, y, u)\gen x +c(x, y, u)\gen y + \varphi(x, y, u)\gen u$.
The residual equivalence group $G^\sim(\gen t)$ consists of invertible transformations of the form
$$t'=t+l,\; x'=X(x, ,y, u),\; y'=Y(x, ,y, u),\; u'=U(x, y, u)$$
 where $l=\text{constant}$. Under such a transformation $Q_2$ is mapped to
$Q'_2=Q_2(X)\gen {x'} + Q_2(Y)\gen {y'} +Q_2(U)\gen {u'}$. We can always choose $X,\, Y$ and $U$ so that $Q_2(U)=1,\; Q_2(X)=Q_2(Y)=0$.
Thus we have the canonical form $\langle \gen t, \gen u\rangle$. If we take $Q_1=\gen u$, then we have $Q_2=a(t)\gen t + b(t,x,y)\gen x + c(t,x,y)\gen y+\varphi(t,x,y)\gen u$.
The residual equivalence group $G^\sim(\gen u)$ is the group of transformations
$$t'=T(t),\, x'=X(t,x,y),\,y'=Y(t,x,y),\,u'=u + U(t,x,y)$$
with $\dot{T}(t)\left(X_xY_y-X_yY_x\right)\neq 0$.
Under such a transformation, $Q_2$ is mapped to
$$Q'_2=a(t)\dot{T}(t)\gen {t'} + Q_2(X)\gen {x'}+ Q_2(Y)\gen {y'}+ [\varphi+Q_2(U)]\gen {u'}.$$
If $a(t)\neq 0$, we may choose $T(t),\, X(t,x,y),\, Y(t,x,y)$ and $U(t,x,y)$ so that
$a(t)\dot{T}(t)=1,\, Q_2(X)=0,\, Q_2(Y)=0,$ and $\varphi+Q_2(U)=0$. Then we may take $Q_1=\gen u,\,
Q_2=\gen t$.  If, however, $a(t)=0$, then $Q_2=b(t, x, y)\gen x +c(t, x, y)\gen y + \varphi(t, x, y)\gen u$
and we then have $b^2+c^2\neq 0$ since we have ${\rm rank}\, A=2$. Then
$$Q'_2=[b(t,x,y)X_x+c(t,x,y)X_y]\gen {x'} +[b(t,x,y)Y_x+c(t,x,y)Y_y]\gen {y'} + [\varphi+b(t,x,y)U_x+c(t,x,y)U_y]\gen {u'}$$
 and because $b^2+c^2\neq0$
we may choose $X(t,x,y),\; Y(t,x,y)$ and $U(t,x,y)$ so that
 $$\varphi+b(t,x,y)U_x+c(t,x,y)U_y=0,\; b(t,x,y)Y_x+c(t,x,y)Y_y=0,\;b(t,x,y)X_x+c(t,x,y)X_y=1$$
 and this gives
 $Q'_2=\gen{x'}$, so we have the canonical form $A=\langle\gen u, \gen x\rangle$.
Hence there are the following admissible,
canonical forms for $\dim A=2$:
 \[ \langle\gen u, x\gen u \rangle,\qquad \langle\gen t, \gen u\rangle,\qquad \langle\gen x, \gen u \rangle. \]
\end{proof}
%---------------------------------------------------------
\begin{theorem}\label{abeliansymm3}
Inequivalent three-dimensional abelian admissible Lie symmetry algebras  are
\begin{eqnarray*}
\langle \gen t, \gen x, \gen u\rangle,\;\; \langle \gen t, \gen u, x\gen u \rangle,~\langle \gen u, x\gen u, y\gen u \rangle,\;\;
\langle \gen x, \gen y, \gen u\rangle,\;\;\langle \gen u, x\gen u, \varphi(t,x)\gen u \rangle,\;  {\rm with}\;\varphi_{xx}\neq0,
\end{eqnarray*}
\end{theorem}
\begin{proof}
For $\dim A=3$ we have $A=\langle Q_1, Q_2, Q_3\rangle$. If $\text{rank}\, A=1$ then we may take $Q_1=\gen u,\;
Q_2=x\gen u$ by the previous argument for $\dim A=2$. We then have $Q_3=\varphi(t,x,y)\gen u$.
If $\varphi_x=\varphi_y=0$
then we find $Q_3=\varphi(t)\gen u$, which $\dim A=3$  gives $\dot{\varphi}(t)\neq0$.
But substituting $\gen u, \varphi(t)\gen u$ in (\ref{sym4}) as symmetries lead to $\dot{\varphi}(t)=0$  that is a contradiction.
If $\varphi_x\neq0,~\varphi_y=0$   then
$\varphi_{xx}\neq0$. Because $\varphi_{xx}=0$ gives  $\varphi(t, x)=\alpha(t)+\beta(t)x$.
Then substituting $\gen u, x\gen u, [\alpha(t)+\beta(t)x]\gen u$ into (\ref{sym4})
we obtain $\dot{\alpha}(t)+ \dot{\beta}(t)x=0$, which is possible only for constant $\alpha(t),\, \beta(t)$.
 But this means that $Q_3$ is a constant
linear combination of $Q_1,$ and $Q_2$, contradicting $\dim A=3$.
So, $\varphi_{xx}(t,x)\neq 0$. For $A=\langle \gen u, x\gen u, \varphi(t, x)\gen u\rangle$ as a
symmetry algebra we find the following form for the  two--dimensional evolution equation:
\begin{eqnarray}\label{eq:lin2}
 u_t=\frac{\varphi_t(t,x)}{\varphi_{xx}(t,x)}u_{xx}+ G(t,x,y)u_{yy}+ H(t,x,y)u_{xy}+ K(t,x,y),
 \end{eqnarray}
with $\varphi_t(t,x)\neq0 , \varphi_{xx}(t,x)\neq 0$. Thus, the evolution
equation is  semi-linear. Now, if $\varphi_y\neq0$, since the residual equivalence group $G^\sim(\gen u, x\gen u)$
is
$$t'=T(t),\, x'=x,\,y'=Y(t,x,y),\,u'=u + U(t,x,y)$$
 with $Y_y\neq 0$, then we can choose $\varphi(t,x,y)=Y(t,x,y)$ and so
$Q_3=y\gen u$ , thus $A=\langle\gen u, x\gen u, y\gen u\rangle$.

We now come to $\dim A=3,\; {\rm rank}\, A=2$. If we take $Q_1=\gen t$ then we may take $Q_2=\gen u$ since
$\langle Q_1, Q_2\rangle$ is a two-dimensional abelian subalgebra of $A$, and there is only one type of
two-dimensional abelian algebra containing $\gen t$ as we have seen. Then we have
$Q_3=a\gen t + b(x, y)\gen x +c(x, y)\gen y + \varphi(x, y)\gen u$ with $a=\,{\rm constant}$ and so we may take
$Q_3=b(x, y)\gen x + c(x, y)\gen y + \varphi(x, y)\gen u$.
The requirement that ${\rm rank}\, A=2$ then gives $b(x, y)= c(x, y)=0$ and we find that $Q_3=\varphi(x, y)\gen u$.
Obviously, $\varphi_x^2+\varphi_y^2\neq0$ since we
must have $\dim A=3$. Now, the residual equivalence group  $G^\sim(\gen t, \gen u)$ is given by
transformations of the form
$$t'=t+k,\; x'=X(x, y),\; y'=Y(x, y),\; u'=u + U(x, y)$$
 with $X_x.Y_y\neq 0$.
 We see that under such a transformation, $Q_3=\varphi(x, y)\gen u$ is transformed to $Q'_3=\varphi(x, y)\gen {u'}$ and
we may take $\varphi(x, y)=X(x, y)$ (or  $\varphi(x, y)=Y(x, y)$) giving
$Q'_3=x'\gen {u'}$ (or $Q'_3=y'\gen {u'}$), so we find the canonical form $A=\langle \gen t, \gen u, x\gen u\rangle$.   If $Q_1=\gen u$,
then we may take   $Q_2=\gen t$ or $Q_2=\gen x$ if $Q_1\wedge Q_2\neq 0$. If $Q_2=\gen t$ then $Q_3=x\gen u$ as
before. If $Q_2=\gen x$ then $Q_3=b(t)\gen x+\varphi(t)\gen u$ because ${\rm rank}\, A=2$. However, putting $Q_1=\gen
u,\; Q_2=\gen x,\; Q_3=b(t)\gen x+\varphi(t)\gen u$ into the (\ref{sym4}) equation, we obtain
$\dot{\varphi}(t)-\dot{b}(t)u_x=0$, which means that $\dot{b}(t)=\dot{\varphi}(t)=0$ and hence $Q_3$ is a linear combination
of $Q_1$ and $Q_2$, which is a contradiction. Thus we cannot have $Q_2=\gen x$ in this case. If $Q_1\wedge Q_2=0$
then $Q_2=\varphi(t,x,y)\gen u$, and since $\langle Q_1, Q_2\rangle$ is a two-dimensional subalgebra of $A=\langle Q_1,
Q_2, Q_3\rangle$ then we know from the above that we may take $Q_2=x\gen u$. This then gives us $Q_3=\gen t$ in
canonical form. Hence we have only the case $A=\langle \gen t, \gen u, x\gen u\rangle$ when $\dim A=3,\; {\rm
rank}\, A=2$.

Now consider the case $\dim A=3,\, {\rm rank}\, A=3$. If we take $Q_1=\gen t$ and $Q_2=\gen u$ then
we may take $Q_3=b(x,y)\gen x+c(x,y)\gen x+\varphi(x,y)\gen u$ with $b^2+c^2\neq0$ since
we have $\gen t\wedge \gen u\wedge Q_3\neq0$. Under $G^\sim(\gen t, \gen u)$ transformation $Q_3$ is mapped
to
 $$Q'_3=[b(x,y)X_x+c(x,y)X_y]\gen {x'} +[b(x,y)Y_x+c(x,y)Y_y]\gen {y'} + [\varphi+b(x,y)U_x+c(x,y)U_y]\gen {u'}$$
 and because $b^2+c^2\neq0$
we can choose $X(x,y),\; Y(x,y)$ and $U(x,y)$ so that
 $\varphi+b(x,y)U_x+c(x,y)U_y=0,\; b(x,y)Y_x+c(x,y)Y_y=0,\;b(x,y)X_x+c(x,y)X_y=1$
 and this gives $Q'_3=\gen{x'}$, so we have the canonical form $A=\langle\gen t, \gen y, \gen x\rangle$.
Now if we take $Q_1=\gen x$ and $Q_2=\gen u$ then we can take $Q_3=a(t)\gen t+b(t,y)\gen x+c(t,y)\gen x+\varphi(t,y)\gen u$
with $a^2+c^2\neq0$ since we have $\gen x\wedge \gen u\wedge Q_3\neq0$. The residual equivalence group $G^\sim(\gen x, \gen u)$
is
 $$t'=T(t),\, x'=x+\xi(t,y),\,y'=Y(t,y),\,u'=u + U(t,y)$$
  with $\dot{T}(t)\neq0$ and $Y_y\neq 0$.
Under this transformation $Q_3$ maps to
\[Q'_3=a(t)\dot{T}(t)\gen {t'}+[a(t)\xi_t+c\xi_y+\varphi]\gen {x'} +[a(t)Y_t+cY_y]\gen {y'} + [a(t)U_t+cU_y+\varphi]\gen {u'}.\]
If $a(t)=0$ then $c\neq0$ and we may choose $\xi, Y$ and $U$ so that $c\xi_y+\varphi=0$, $cY_y=1$ and $cU_y+\varphi=0$.
Thus we find $A=\langle \gen x, \gen u, \gen y\rangle$. The $a(t)\neq0$ doesn't lead to a new case.
\end{proof}
%--------------------------------------------------------------------------------------------------------------
\begin{theorem}\label{abeliansymm4}
Inequivalent  abelian admissible Lie symmetry algebras $A$ with ${\rm dimension}\; A\geq 4$  are
\begin{eqnarray*}
\langle \gen t, \gen u, x\gen u, y\gen u \rangle,\;\;\langle \gen t, \gen x, \gen u, y\gen u \rangle,
\langle \gen t, \gen x, \gen u, y\gen x \rangle,\;\;\langle \gen t, \gen x, \gen y, \gen u \rangle,\\
\langle \gen u, x\gen u, y\gen u, \varphi(t,x,y)\gen u, q_1(t,x,y)\gen u,\dots, q_k(t,x,y)\gen u \rangle,
\;{\rm with}\;\;\varphi_{xx}\neq 0,\,(q_i)_{xx}\neq 0,\\
\langle \gen t, \gen u, x\gen u, y\gen u, \varphi(x,y)\gen u, q_1(x,y)\gen u,\dots, q_k(x,y)\gen u \rangle,\;\;
{\rm with}\;\;\varphi_{xy}\neq 0,\,(q_i)_{xy}\neq 0.
\end{eqnarray*}
\end{theorem}
\begin{proof}
Now consider the case $\dim A=4,\, {\rm rank}\, A=4$. We have $A=\langle Q_1, Q_2, Q_3, Q_4\rangle$ and
$$Q_i=a_i(t)\gen t + b_i(t,x,y,u)\gen x + c_i(t,x,y,u)\gen y+\varphi_i(t,x,y,u)\gen u$$
 for $i=1, 2, 3$.
Since $Q_1\wedge Q_2\wedge Q_3\wedge Q_4\neq 0$
at least one of the coefficients $a_i(t)$  must be different from zero. So, assume $a_1(t)\neq 0$. In this case,
we may take $Q_1=\gen t$ in canonical form, according to Theorem \ref{canonicalform}. Then we must have
$Q_i=a_i\gen t + b_i(x,y,u)\gen x + c_i(x,y,u)\gen y+\varphi_i(x,y,u)\gen u$ for $i=2, 3, 4$.
We may take $a_i=0$ for $i=2, 3, 4$ because of
linear independence of $Q_1, Q_2, Q_3, Q_4$.
Hence
$$Q_2=b_2(x,y,u)\gen x + c_2(x,y,u)\gen y+\varphi_2(x,y,u)\gen u, Q_3=b_3(x,y,u)\gen x + c_3(x,y,u)\gen y+\varphi_3(x,y,u)\gen u$$
and $Q_4=b_4(x,y,u)\gen x + c_4(x,y,u)\gen y+\varphi_4(x,y,u)\gen u$.
Consequently, we may, by Theorem \ref{canonicalform} take $Q_2=\gen x$ and $Q_3=\gen u$.
Finally, this means that $Q_4=b_4(y)\gen x + c_4(y)\gen y+\varphi_4(y)\gen u$ by commutativity,
and $c_4(y)\neq 0$ because $\gen t\wedge \gen x\wedge \gen u\wedge Q_4\neq 0$.
The residual equivalence group $G^\sim(\gen t, \gen x, \gen u)$ is given by transformations of the form
$$t'=t+k,\; x'=x + Y(y),\; y'=Y(y),\; u'=u+U(y)$$ with $Y_y\neq 0$.
Under such a transformation, $Q_4$ is transformed to
$$Q'_4=[b_4(y)+c_4(y)X_y]\gen {x'} + c_4(y)Y_y\gen {y'}+[\varphi_4(y)+c_4(y)U_y]\gen {y'}.$$
Since $c_4(y)\neq 0$ then we may choose $X(y), Y(y)$ and $U(y)$ so
that $b_4(y)+c_4(y)X_y=0, c_4(y)Y_y=1$ and $\varphi_4(y)+c_4(y)U_y=0$, so we may take $Q_4=\gen y$ in
canonical form. Consequently, we have only one canonical form $A=\langle \gen t, \gen x, \gen u, \gen y\rangle$.

For $\dim A\geq5,\, {\rm rank}\, A=4$, we may choose $Q_1=\gen t, Q_2=\gen x, Q_3=\gen u, Q_4=\gen y$. For any
other symmetry $Q=a(t)\gen t + b(t, x, y,u)\gen x + c(t, x, y,u)\gen y+\varphi(t, x, y,u)\gen u$, the coefficients
must be constants, so that $Q$ will be a constant linear combination of $Q_1, Q_2, Q_3, Q_4$ contradicting with
$\dim A\geq5$.

For $\dim A= 4,\,{\rm rank}\, A=3$, we may take $Q_1=\gen t,~ Q_2=\gen x,~ Q_3=\gen u$
and for any other symmetry we take $Q_4=b(y)\gen x+\varphi(y)\gen u$.
The residual equivalence group $G^\sim(\gen t, \gen x, \gen u)$
is
$$t'=t+k,\, x'=x+\xi(y),\,y'=Y(y),\,u'=u + U(y)$$
 with $Y_y\neq0$.
Since $\dim A= 4$ then we have $b_y^2+\varphi_y^2\neq0$. If $b_y=0,~\varphi\neq0$, since $Y_y\neq0$
thus we may choose $\varphi(y)=Y(y)$ and then we obtain $Q_4=y\gen u$.
Therefore we find $A=\langle \gen t, \gen x, \gen u, y\gen u\rangle$.
Now if $b_y\neq0,~\varphi=0$,  we may choose $b(y)=Y(y)$ and then we find $Q_4=y\gen x$ and thus
 we have $A=\langle \gen t, \gen x, \gen u, y\gen x\rangle$. We may $Q_1=\gen x,~ Q_2=\gen y,~ Q_3=\gen u$
and any other symmetry is $Q_4=b(t)\gen x+c(t)\gen y+\varphi(t)\gen u$, because $A$ is a abelian
Lie algebra with $\dim A= 4$. Putting $Q_1, Q_2, Q_3, Q_4$ in determining equations (\ref{determin})
leads to $\dot{\varphi}(t)-\dot{b}(t)u_x-\dot{c}(t)u_y=0$. Therefore we find
$\dot{\varphi}(t)=\dot{b}(t)=\dot{c}(t)=0$, that means $Q_4$ is a Linear combination of
$Q_1, Q_2, Q_3$ contradicting with $\dim A= 4$.
 Consequently we have only two canonical forms $A=\langle \gen t, \gen x, \gen u, y\gen u\rangle$
 and  $A=\langle \gen t, \gen x, \gen u, y\gen x\rangle$ for an
abelian Lie algebra with $\dim A=4,\, \text{rank}\, A=3$.

Now we come to $\dim A\geq 4, {\rm rank}\, A=3$. We may take $Q_1=\gen t, Q_2=\gen x, Q_3=\gen u$.
For any other symmetry vector field $Q=a(t)\gen t + b(t,x,u)\gen x + c(t,x,u)\gen u$ the coefficients must be
constants, so that $Q$ will be a constant linear combination of $Q_1, Q_2, Q_3$, contradicting $\dim A\geq 4$. So
there are no admissible abelian Lie algebras $A$ with $\dim A\geq 4,\, {\rm rank}\, A=3$.

For $\dim A= 4,\,{\rm rank}\, A=2$ we may take $Q_1=\gen t,\, Q_2=\gen u,\, Q_3=x\gen u$, and for $Q_4$ we take
$Q_4=\varphi(x,y)\gen u$ because $A$ is abelian and ${\rm rank}\, A=2$. Further $\varphi_y\neq0$, for otherwise $Q_4$ would
be a constant linear combination of $Q_2$ and $Q_3$, contradicting $\dim A=4$. Thus we have $A=\langle \gen t,
\gen u, x\gen u, \varphi(x,y)\gen u\rangle$ with $\varphi_y\neq0$.
The residual equivalence group  $G^\sim(\gen t, \gen u, x\gen u)$ is given by
transformations of the form
 $$t'=t+k,\; x'=x,\; y'=Y(x, y),\; u'=u + U(x, y)$$
  with $Y_y\neq 0$.
Under this transformation $Q_4$ maps to $Q'_4=\varphi(x,y)\gen u'$ and
since $\varphi_y\neq0$ we may choose $\varphi(x,y)=Y(x,y)$.
Thus $Q_4=y\gen u$ and $A=\langle \gen t, \gen u, x\gen u, y\gen u\rangle$.

For $\dim A= 5,\; \text{rank}\, A=2$ we find that any additional symmetry operator $Q$ must be of the form
$Q=\varphi(x,y)\gen u$ with $\varphi_y\neq 0$ and, substituting the vector fields a symmetries in the (\ref{determin}) equations
gives $\varphi_{xy}K+\varphi_{yy}G+\varphi_{xx}F=0$ with $\varphi_{xy}^2+\varphi_{yy}^2+\varphi_{xx}^2\neq0$.
From this it follows that $\varphi_x\neq0, \varphi_y\neq0$ and then $\varphi_{xy}\neq0$. Hence we find
$A=\langle \gen t, \gen u, x\gen u, y\gen u, \varphi(x,y)\gen u\rangle$
with $\varphi_{xy}\neq0$.
For $\dim A\geq 5,\; \text{rank}\, A=2$, arguments similar to those for $\dim A=5,\,\text{rank}\, A=2$ give us
abelian Lie algebras
\begin{eqnarray}\label{eq:alg}
A=\langle \gen t, \gen u, x\gen u, y\gen u, \varphi(x,y)\gen u, q_1(x,y)\gen u, \ldots, q_k(x,y)\gen u\rangle
\end{eqnarray}
with the sole proviso that each of the $q:$ s satisfy $q_{xy}\neq 0$ and
\[\frac{\varphi_{xx}F+\varphi_{yy}G}{\varphi_{xy}}=\frac{q_{xx}F+q_{yy}G}{q_{xy}}, \]
as well as linear independence of the vector fields (\ref{eq:alg}).
We obtain equations of the form
\begin{eqnarray}\label{eq:lin3}
 u_t=F(x,y)u_{xx} + G(x,y)u_{yy}+H(x,y)u_{xy}-\frac{\varphi_{yy}G+\varphi_{xx}F}{\varphi_{xy}},
\end{eqnarray}
with $\varphi_{xy}\neq 0,\; \varphi_{xx}^2+\varphi_{yy}^2\neq0$.

For $\dim A\geq 4,\; \text{rank}\, A=1$, arguments similar to those for $\dim A=3,\,\text{rank}\, A=1$ give us
abelian Lie algebras $A$ of the form
\begin{eqnarray}\label{alg}
A=\langle \gen u, x\gen u, y\gen u, \varphi(t,x,y)\gen u, q_1(t,x,y)\gen u,\dots, q_k(t,x,y)\gen u \rangle
\end{eqnarray}
with the sole proviso that each of the $q:$ s satisfy $q_{xx}(t,x, y)\neq 0, \varphi(t,x,y)\neq0$ and
\[\frac{\varphi_t-\varphi_{yy}G-\varphi_{xy}K}{\varphi_{xx}}=\frac{q_t-q_{yy}G-q_{xy}K}{q_{xx}}, \]
as well as linear independence of the vector fields (\ref{alg}).
We obtain equations of the form
\begin{eqnarray}\label{eq:lin4}
 u_t=\left[\frac{\varphi_t-\varphi_{yy}G-\varphi_{xy}K}{\varphi_{xx}}\right]u_{xx} + G(t,x,y)u_{yy}+H(t,x,y)u_{xy} +K(t,x,y).
\end{eqnarray}
\end{proof}
%----------------------------------------------------------------
We have the following corollary of the above results:
\begin{corollary}\label{linearizableeqns}
The evolution equations of the form (\ref{eveqn}) admitting abelian Lie algebras $\mathsf{A}$
with $\dim\mathsf{A}\geq 5$ or if $\dim\mathsf{A}\geq3,\;\; {\rm rank}\,\mathsf{A}= 1$
as symmetries are linearizable.
\end{corollary}
\begin{proof}
The proof is just by calculation: for $\dim\mathsf{A}\geq 3,\;{\rm
rank}\,\mathsf{A}=1$ we find the evolution equations (\ref{eq:lin2}) and (\ref{eq:lin4}) which is linear in $u$ and its derivatives.
Similarly, the rank-two algebra
\[A=\langle \gen t, \gen u, x\gen u, y\gen u, \varphi(x,y)\gen u, q_1(x,y)\gen u,\dots, q_k(x,y)\gen u \rangle,\;\;
{\rm with}\;\;\varphi_{xy}\neq 0,\,(q_i)_{xy}\neq 0,\]
gives the (\ref{eq:lin3}) equation.
\end{proof}
%-----------------------------------------------------------------
 We have the following table of non-linear equations form of \eqref{eveqn} admitting abelian Lie algebras as
symmetries:
\begin{center}
\begin{tabular}{ |p{2.3cm}| p{12.5cm}| }
\hline
\multicolumn{2}{|c|}{The invariant solutions of symmetries admitting abelian Lie algebras for non-linear forms of \eqref{eveqn}} \\
\hline
Realizations & invariant solutions \\
\hline
$\langle \gen t\rangle$ & $u_t=\widetilde{F}(x,y,u,u_x, u_y)u_{xx} + \widetilde{G}(x,y, u, u_x, u_y)u_{yy}+\widetilde{H}(x,y,u,u_x, u_y)u_{xy}+\widetilde{K}(x,y,u,u_x, u_y)$ \\ \hline
$\langle \gen u\rangle$ & $ u_t=\widetilde{F}(t,x,y,u_x, u_y)u_{xx} + \widetilde{G}(t,x,y,u_x, u_y)u_{yy}+\widetilde{H}(t,x,y,u_x, u_y)u_{xy}+\widetilde{K}(t,x,y,u_x, u_y)$ \\ \hline
$\langle \gen t, \gen u\rangle$ & $ u_t=\widetilde{F}(x,y,u_x, u_y)u_{xx} + \widetilde{G}(x,y,u_x, u_y)u_{yy}+\widetilde{H}(x,y,u_x, u_y)u_{xy}+\widetilde{K}(x,y,u_x, u_y)$ \\ \hline
$\langle \gen x, \gen u\rangle$    &    $u_t=\widetilde{F}(t,y,u_x, u_y)u_{xx} + \widetilde{G}(t,y,u_x, u_y)u_{yy}+\widetilde{H}(t, y, u_x, u_y)u_{xy}+\widetilde{K}(t,y,u_x, u_y)$   \\ \hline
$\langle \gen u, x\gen u\rangle$    &    $u_t=\widetilde{F}(t,x,y,u_y)u_{xx} + \widetilde{G}(t,x,y,u_y)u_{yy}+\widetilde{H}(t,x,y,u_y)u_{xy}+\widetilde{K}(t,x,y,u_y)$   \\ \hline
$\langle \gen t, \gen x, \gen u\rangle$   &   $u_t=\widetilde{F}(y,u_x, u_y)u_{xx} + \widetilde{G}(y,u_x, u_y)u_{yy}+\widetilde{H}(y,u_x, u_y)u_{xy}+\widetilde{K}(y,u_x, u_y)$    \\ \hline
$\langle \gen t, \gen u, x\gen u\rangle$   &   $u_t=\widetilde{F}(x,y,u_y)u_{xx} + \widetilde{G}(x,y,u_y)u_{yy}+\widetilde{H}(x,y,u_y)u_{xy}+\widetilde{K}(x,y,u_y)$     \\ \hline
$\langle \gen x, \gen y, \gen u\rangle$      &   $u_t=\widetilde{F}(t,u_x,u_y)u_{xx} + \widetilde{G}(t,u_x,u_y)u_{yy}+\widetilde{H}(t,u_x,u_y)u_{xy}+\widetilde{K}(t,u_x,u_y)$     \\ \hline
 $\langle \gen t, \gen x, \gen u, y\gen u \rangle$    &   $ u_t=\widetilde{F}(y,u_x)u_{xx} + \widetilde{G}(y,u_x)u_{yy}+\widetilde{H}(y,u_x)u_{xy}+\widetilde{K}(y,u_x)$     \\ \hline
$\langle \gen t, \gen x, \gen u, y\gen x \rangle$       &  $  u_t=\left[{u_y^2\over u_x}\widetilde{G}(y,u_x)-{u_y\over u_x}\widetilde{\widetilde{G}}(y,u_x)+\widetilde{\widetilde{\widetilde{G}}}(y,u_x)\right]u_{xx} + \widetilde{G}(x,y,u_y)u_{yy}
+\widetilde{H}(x,y,u_y)u_{xy}+{2u_y\over u_x}\widetilde{G}(y,u_x)+\widetilde{\widetilde{G}}(y,u_x) $     \\ \hline
$\langle \gen t, \gen x, \gen y, \gen u \rangle$      &  $u_t=\widetilde{F}(u_x,u_y)u_{xx} + \widetilde{G}(u_x,u_y)u_{yy}+\widetilde{H}(u_x,u_y)u_{xy}+\widetilde{K}(u_x,u_y)$     \\
\hline
\end{tabular}
\end{center}
%---------------------------------------------------------------------------
\section{Conclusion and future work}
We have carried out the abelain Lie symmetry classification of evolution equation \eqref{eveqn}.  There are just two inequivalent forms for one--dimensional
abelian Lie symmetry algebras  $\langle \gen t\rangle, \; \langle \gen u\rangle$ and  three inequivalent forms for two--dimensional
$\langle \gen t, \gen u \rangle, \; \langle \gen x, \gen u\rangle, \; \langle \gen u, x\gen u\rangle$. For three--dimensional we have
$\langle \gen t, \gen x, \gen u\rangle,$ $\langle \gen t, \gen u, x\gen u \rangle,$ $\langle \gen u, x\gen u, y\gen u \rangle,$ $\langle \gen x, \gen y, \gen u\rangle,$ $\langle \gen u, x\gen u, \varphi(t,x)\gen u \rangle$  with $\varphi_{xx}\neq0$ abelian Lie symmetry algebras. We have four inequivalent   four--dimensional admissible abelian Lie symmetry algebras
$\langle \gen t, \gen u, x\gen u, y\gen u \rangle,$ $\langle \gen t, \gen x, \gen u, y\gen u \rangle,$  $\langle \gen t, \gen x, \gen u, y\gen x \rangle,$ $\langle \gen t, \gen x, \gen y, \gen u \rangle$
and for  admissible abelian Lie symmetry algebras with dimension bigger than four are
\begin{eqnarray*}
\langle \gen u, x\gen u, y\gen u, \varphi(t,x,y)\gen u, q_1(t,x,y)\gen u,\dots, q_k(t,x,y)\gen u \rangle,
\;{\rm with}\;\;\varphi_{xx}\neq 0,\,(q_i)_{xx}\neq 0,\\
\langle \gen t, \gen u, x\gen u, y\gen u, \varphi(x,y)\gen u, q_1(x,y)\gen u,\dots, q_k(x,y)\gen u \rangle,\;\;
{\rm with}\;\;\varphi_{xy}\neq 0,\,(q_i)_{xy}\neq 0.
\end{eqnarray*}
The evolution equations of the form (\ref{eveqn}) admitting five-dimensional abelian Lie algebras  or with a dimension higher than  five
or if with three-dimensional  abelian Lie algebras  or with a dimension higher than three with rank-one as symmetries are linearizable.

In future work,  we show the evolution equations  \eqref{eveqn} admit ${\frak s}{\frak l}(2, {\Bbb R})$ as semi-simple symmetry algebras we classify these realizations.
%---------------------------------------------------------------
\section{Acknowledgement}
I thank the Mathematics Department, Link{\"o}ping University and professor  Peter Basarab-Horwath for their hospitality during my stay in Link{\"o}ping, where this work was started. 
The author also wishs to thank Professor  R. O. Popovych for
useful comments. 

This work was supported by the Babol Noshirvani University of Technology under
Grant No. BNUT/391024/99.

%%%%%%%%%%%%%%%%%%%%%%%%%%%%%%%%%%%%%%%%%%
\bibliographystyle{amsplain}
%\bibliography{toriExoticGeometry}

\begin{thebibliography}{10}
%
\bibitem{Basarab-Horwath2001}
{P}. {Basarab-Horwath}, {V}. {Lahno}, and {R}. {Zhdanov}.
\newblock {The} stucture of Lie algebras and the classification problem for partial differential equations.
\newblock {\em Acta Appl. Math.}, 69:43--94, 2001.
%\newblock \href{http://arXiv.org/math-ph/0005013}{math-ph/0005013}.

\bibitem{Basarab-Horwath2017}
{P}. {Basarab-Horwath}, and {F}.~{G{\"u}ng{\"o}r}.
\newblock {Linearizability} for third order evolution equations.
\newblock {\em J. Math. Phys.}, 58: 081507, 2017.

\bibitem{Basarab-Horwath2013}
{P}. {Basarab-Horwath}, {F}. {G{\"u}ng{\"o}r}, and {V}. {Lahno}.
\newblock {The} Symmetry classification of
third-order nonlinear evolution equations. {Part} {I}: {Semi-simple} {Algebras}.
\newblock {\em Acta Appl. Math.}, 124: 123--170, 2013.

\bibitem{Popovych2017}
{A}. {Bihlo} and {R}. {O}. { Popovych}.
\newblock {Group} classification of linear evolution equations.
\newblock {\em J. Math. Anal. Appl. }, 448(2): 982--1005, 2017.

\bibitem{Popovych2020}
{A}. {Bihlo}, and {R}. {O}. { Popovych}.
\newblock {Zeroth} order conservation laws of two dimensional shallow water equations
with variable bottom topography.
\newblock {\em Stud Appl Math.}, 145:291--321, 2020.

\bibitem{Demetriou2008}
{E}. {Demetriou}, {E}. { Ivanova}, and  {N}. {M}. {Sophocleous}.
\newblock {Group} analysis of $(2 + 1)$-- and $(3 + 1)$--dimensional diffusion--convection equations.
\newblock {\em J. Math. Anal. Appl.}, 348: 55--65, 2008.

\bibitem{Gazeau92}
J.P. Gazeau and P.~Winternitz.
\newblock {Symmetries} of {Variable} {Coefficient} {Korteweg}-de {Vries}
  {Equations}.
\newblock {\em J. Math. Phys.}, 33(12):4087-4102, 1992.

\bibitem{Gungor96}
F.~G{\"u}ng{\"o}r, M.~Sanielevici, and P.~Winternitz.
\newblock {Equivalence} {Classes} and {Symmetries} of the {Variable} {Coefficient}
{Kadomtsev--Petviashvili} {Equation}.
\newblock {\em Nonlinear Dynamics},  35: 381--396, 2004.

\bibitem{Gungor2004}
{F}.~{G{\"u}ng{\"o}r}, {V}.~{Lahno}, and {R}.~{Zhdanov}.
\newblock Symmetry classification of KdV-type nonlinear evolution equations.
\newblock {\em J. Math. Phys.}, 45:2280-2313, 2004.

\bibitem{Kurujyibwami2018}
{C}. {Kurujyibwami}, {P}. {Basarab-Horwath}, and {R}. {O}. { Popovych}.
\newblock Algebraic Method for Group Classification of (1+1)-Dimensional Linear Schr{\"o}dinger Equations.
\newblock {\em Acta Appl. Math.}, 157: 171-203, 2018.

\bibitem{Kontogiorgis2018}
{S}.  {Kontogiorgis}, {R}. {O}.  {Popovych}, and {C}. {Sophocleous},
\newblock Enhanced Symmetry Analysis of Two-Dimensional Burgers System
\newblock {\em Acta Appl. Math.}, 163: 91--128, 2019.

\bibitem{Kurujyibwami2020}
{C}. {Kurujyibwami}, and {R}. {O}. { Popovych}.
\newblock Equivalence groupoids and group classification of multidimensional nonlinear Schr{\"o}dinger equations.
\newblock {\em J. Math. Anal. Appl.}, 491: 124271, 2020.

\bibitem{Lahno05}
{V}.~{I}. {Lahno}, and {R}.~{Z}. {Zhdanov}.
\newblock Group classification of  nonlinear wave equations.
\newblock {\em J. Math. Phys.}, 46:1-37, 2005.

\bibitem{Olver91}
P. J. Olver.
\newblock {\em Applications of Lie Groups to Differential Equations}.
\newblock Springer, New York, 1991.

\bibitem{Rajaee2008}
{L}. {Rajaee}., {H}. {Eshraghi} and {R}. {O}. { Popovych}.
\newblock {Multi--dimensional} quasi-simple waves in weakly dissipative flows.
\newblock {\em Physica D},  237: 405-419, 2008.

\bibitem{Bihlo2012}
{Stanislav}. {Opanasenko}.,{ Alexander}. {Bihlo}., and {R}. {O}. {Popovych}.,
\newblock Equivalence groupoid and group classification of a class of variable-coefficient Burgers equations,
\newblock {\em Journal of Mathematical Analysis and Applications }  491(1):124215, 2020.

\bibitem{Vaneeva2020}
{O}. {O}. {Vaneeva}., {A}. {Bihlo} and {R}. {O}. { Popovych}.
\newblock Generalization of the algebraic method of group classification with application to
 nonlinear wave and elliptic equations.
\newblock {\em Commun Nonlinear Sci Numer Simulat},  91: 105419, 2020.

\bibitem{Zhdanov99}
{R}.~{Z}. {Zhdanov} and {V}.~{I}. {Lahno}.
\newblock {Group} {classification} of {heat} {conductivity} {equations} with a
  {nonlinear} {source}.
\newblock {\em J. Phys. A : Math. and Gen.}, 32:7405--7418, 1999.
%\newblock \href{http://arXiv.org/math-ph/0005013}{math-ph/0005013}.

\bibitem{Zhdanov07}
{R}.~{Z}. {Zhdanov} and {V}.~{I}. {Lahno}.
\newblock Group classification of the general second-order evolution
          equation: semi-simple invariance groups.
\newblock {\em J. Phys. A: Math. and Theor.}, 40:5083--5103, 2007.

\end{thebibliography}

\end{document}